\documentclass[11pt]{article}
\usepackage{thumbpdf,lmodern}

\usepackage{framed}

\usepackage{xr-hyper}
\usepackage{amsthm}
\usepackage{hyperref}
\usepackage{amsmath}
\usepackage{amssymb}
\usepackage{amscd}
\usepackage{graphics}
\usepackage{graphicx}
\usepackage{latexsym}
\usepackage{stmaryrd}
\usepackage[all,cmtip]{xy}

\usepackage{xcolor}

\usepackage[round, sort]{natbib}
\usepackage{blindtext}

\usepackage[ruled,vlined]{algorithm2e}

\usepackage{microtype}
\usepackage{multicol}

\newtheorem{theorem}{Theorem}[subsection]
\newtheorem{proposition}[theorem]{Proposition}
\newtheorem{lemma}[theorem]{Lemma}

\newtheorem{corollary}[theorem]{Corollary}

\theoremstyle{definition}

\theoremstyle{remark}
\newtheorem{remark}[theorem]{Remark}

\numberwithin{equation}{subsection}

\usepackage{algpseudocode,float}

\author{Joseph Ross \thanks{Splunk Inc., 3098 Olsen Dr, San Jose, CA 95128; josephr@splunk.com, joseph.ross@gmail.com}}

\title{Asymmetric scale functions for $t$-digests}
\date{}

\begin{document}
\maketitle

\begin{abstract}
The $t$-digest is a data structure that can be queried for approximate quantiles, with greater accuracy near the minimum and maximum of the distribution.
We develop a $t$-digest variant with accuracy asymmetric about the median,
thereby making possible
alternative tradeoffs between computational resources and accuracy which may be of particular
interest for distributions with significant skew.
After establishing some theoretical properties of scale functions for $t$-digests, we show that a tangent line construction on the familiar scale functions preserves the crucial properties that allow $t$-digests to operate online and be mergeable. We conclude with an empirical study demonstrating the asymmetric variant preserves accuracy on one side of the distribution with a much smaller memory footprint.
\end{abstract}
\section{Introduction}

Recently the $t$-digest \citep{dunning2019computing} has gained prominence as an efficient data structure for online estimation of quantiles of large data streams. 
The digest consists of a collection of weighted centroids on the real line, with the weight representing cluster size (the number of observations near the corresponding centroid).
In comparison to other methods, the $t$-digest is notable for its ability to have variable accuracy in different regions of quantile space. The accuracy is controlled by a scale function, which governs the permissible compression (expressed as a bound on cluster size, see \citep{dunning2019size}) as a function of the quantile $q$. Using a linear scale function turns the $t$-digest into a dynamic version of a histogram with equal-sized bins, but using logarithmic or (inverse) trigonometric functions allows the digest to achieve greater accuracy near the tails (i.e., $q$ near $0$ or $1$) and comparatively less accuracy near the median ($q = {1 \over 2}$).

The capacity of the $t$-digest to operate online imposes a requirement on the scale function, namely that a collection of centroids compatible with a given scale function remains so when new samples are inserted. Since forming the ordered union of two digests may be described as a sequence of insertions from the viewpoint of either digest, meeting this requirement also implies $t$-digests can be merged to form a new one that inherits the accuracy bounds of its constituents and thus large datasets may $t$-digested in parallel \citep[\S 2.5]{dunning2019computing}. For the well-known scale functions, preservation of the constraint under insertion is proved in \citep{dunning2019conservation}.

\textbf{Motivation.}
All of the scale functions the author is aware of are symmetric about $q = {1 \over 2}$, and thus expend similar computational resources on those parts of the distribution near $q=0$ and those near $q=1$. 
In practice there are scenarios in which one tail of the distribution carries considerably more excitement than the other. For example, in application performance monitoring, the latency of individual operations or execution paths is often distributed with significant positive skew, as the overwhelming majority of executions complete quickly and uneventfully, while a relatively small number of outlying executions exhibits greater variation.
In practical terms, the difference between a 97th percentile and a 99th percentile operation execution is greater than the difference between a 3rd percentile and a 1st percentile execution, and so accuracy near $q=1$ is ``worth more" than accuracy near $q=0$.
In \cite[esp.~\S 3.2]{sampler} we have described a tail-based sampling method for distributed traces that requires a compact device for approximating quantiles and ranks, and in this setting we would like to make fine-grained distinctions near $q=1$, whereas very little of our budget will be devoted to keeping execution traces near $q=0$ in any case.

A related context is monitoring service level objectives in distributed computing environments \citep[Ch.~4]{beyer2016site}, \citep{sloss2017calculus}: it is common to treat upper quantiles of request latency as service level indicators \citep[Ch.~4]{beyer2016site}, which may be implemented as a client querying a $t$-digest for a particular quantile value near $q=1$ (but not near $q=0$). While it may not be possible to enhance the resolution of a $t$-digest exactly in a neighborhood of a specified quantile (since insertions may shift a region of data for which only a coarse summary is available into a region in which greater accuracy is required), an asymmetric scale function allows one to strike a better balance between computational resources and accuracy (e.g., save computational resources without compromising the accuracy of the required estimate, or increase accuracy for the required estimate by using an asymmetric scale function with a larger $\delta$ parameter). Especially for high-volume endpoints over longer time windows, the asymmetric $t$-digests we propose here are a natural family of data structures on which to base approximate calculations.

\textbf{Contributions.} In this paper, we prove (Subsection \ref{piecewise}) that a simple modification of the common scale functions continues to enjoy the preservation of the constraint under insertion property. The construction uses a piecewise definition in which we keep the scale function for $q \in (p, 1]$ and use the best linear approximation of the scale function at $p$ (i.e., the function whose graph is the tangent line to the graph of the scale function at $p$) for $q \in [0, p]$. Our approach is motivated by some brief theory (Section \ref{theory}), from which we conclude that decent scale functions must be differentiable, and from which we deduce an explicit criterion for verifying the decency of a candidate scale function. As a consequence we analyze the case of polynomial scale functions (Subsection \ref{polynomials}). We conclude with some empirical results in Section \ref{results}.

\textbf{Previous work.} For some background on other methods for computing quantiles in an online fashion, we refer the reader to \citep[1.1]{dunning2019computing}, in particular the Q-digest of \citep{shrivastava2004medians}, and the works of 
\citep{munro1980selection}, \citep{chen2000incremental}, and \citep{greenwald2001space}.
The moment-based quantile sketch has recently emerged as another compact data structure for quantile estimation \citep{gan2018moment}.

\section{Generalities} \label{theory}

\subsection{Definitions}

An ordered set of clusters $\mathcal{C} := \{ {C}_1, \ldots , {C}_n \}$ on a set of points in $\mathbb{R}$ is called a $t$-digest with respect to a scale function $k : [0, 1] \to \mathbb{R}$ if every cluster has unit weight or satisfies $k(q_{right}) - k(q_{left}) \leq 1$ \citep[\S 2.1]{dunning2019computing}.
The quantity $k(q_{right}) - k(q_{left})$ is called the $k$-size of the cluster.
We will always require $k$ to be non-decreasing and piecewise differentiable.

We will be interested in the operation of inserting a collection of samples $\Delta$ into a given set $\cal{C}$; denote the result by ${\cal{C}} \cup \Delta$. The notation does not specify where $\Delta$ was inserted.
We say a scale function $k$ \textit{accepts insertions} (or is \textit{insertion-accepting}) if given any $t$-digest $\mathcal{C}$ with respect to $k$, every cluster $C_i \in \mathcal{C}$ continues to have $k$-size less than or equal to $1$ when its quantile range is calculated in ${\cal{C}} \cup \Delta$.

As the condition $k(q_{right}) - k(q_{left}) \leq 1$ indeed implies a scale for $k$, it is natural to restrict our attention to insertion-accepting scale functions with the property that $\delta k$ is again insertion-accepting for any $\delta > 0$. We call such insertion-accepting scale functions \textit{decent}.

If the insertion is to the left of a cluster spanning $[q_1, q_2]$ in $\mathcal{C}$, the cluster spans $[\alpha + (1 - \alpha) q_1, \alpha + (1 - \alpha) q_2]$ in ${\cal{C}} \cup \Delta$, where $0 < \alpha < 1$ is the proportion represented by $\Delta$ in ${\cal{C}} \cup \Delta$ (i.e., $|\Delta | / (|\Delta| + |\cal{C}|))$. When the insertion is to the right, the cluster spans $[(1 - \alpha) q_1, (1 - \alpha) q_2]$ in ${\cal{C}} \cup \Delta$.

\subsection{Characterizations}

\begin{lemma} \label{two_vble_char}
The scale function $k$ is decent if and only if for all $0 < q_1 < q_2 < 1$ and all $\alpha \in (0,1)$, we have $k(q_2') - k(q_1') \leq k(q_2) - k(q_1)$ for $(q_1', q_2') =  (\alpha + (1 - \alpha) q_1, \alpha + (1 - \alpha) q_2)$ and for $(q_1', q_2') = ((1 - \alpha) q_1, (1 - \alpha) q_2)$.

\end{lemma}
\begin{proof}
Clearly the condition implies $k$ accepts insertions. Since the condition is preserved under scaling by $\delta > 0$, the condition implies $\delta k$ accepts insertions, i.e., $k$ is decent.

If $k(q_2') - k(q_1') > k(q_2) - k(q_1)$ for some $q_1, q_2, \alpha$, we can find $\delta > 0 $ such that $\delta (k(q_2') - k(q_1') ) > 1 >  \delta( k(q_2) - k(q_1))$. An insertion into a set of clusters realizing the transformation $(q_1, q_2) \mapsto (q_1', q_2')$ would then violate the insertion-accepting condition for $\delta k$, and so decency implies the condition.
\end{proof}

Rearranging the inequality of the preceding lemma gives the following characterization of decent scale functions.

\begin{corollary} \label{one_vble_char}
The scale function $k$ is decent if and only if for all $\alpha \in (0, 1)$, the functions $k((\alpha + (1 - \alpha) q) - k(q)$ and $k((1 - \alpha) q) - k(q)$ are non-increasing on $[0,1]$.
\end{corollary}

\subsection{Properties}

\begin{lemma} \label{cone}
Decent scale functions form a convex cone: if $k_1 ,k_2$ are decent, then so is $\delta_1 k_1 + \delta_2 k_2$ for any $\delta_1, \delta_2 > 0$.
\end{lemma}

\begin{proof}
Use the characterization of Corollary \ref{one_vble_char} or that of Lemma \ref{two_vble_char}.
\end{proof}

\begin{lemma}
A decent scale function is continuous.
\end{lemma}
\begin{proof} Let $q^* \in (0, 1)$ be a point where continuity fails, and let $k_{left}(q^*) \neq k_{right}(q^*)$ denote the left and right hand limits of $k$ at $q^*$. By piecewise continuity, we can find a pair of points $q_1, q_2 < q^*$ such that $k(q_2) - k(q_1) < k_{right}(q^*) - k_{left}(q^*)$. For an insertion pushing $q_2$, but not $q_1$, across the point of discontinuity, we have $q_2' > q^* > q_1'$ and so $k(q_2') - k(q_1') \geq k_{right}(q^*) - k_{left}(q^*)$. Combining the inequalities produces a violation of the condition of Lemma \ref{two_vble_char}.
\end{proof}

\begin{proposition} \label{differentiable}
A decent scale function is differentiable.
\end{proposition}
\begin{proof}
Let $q^* \in (0, 1)$ be a point where differentiability fails, suppose $k'_{left}(q^*)$ and $k'_{right}(q^*)$ both exist, and suppose $k'_{left}(q^*) < k'_{right}(q^*)$. In this case we shift a centroid to the right; if the inequality were reversed, we would shift it to the left.

Let $q_n \to q^*$ be a sequence approaching $q^*$ from below, and define $\alpha_n := (q^* - q_n) / (1 - q_n)$ and $q_n' := \alpha_n + ( 1 - \alpha_n) q^*$. Note that $\alpha_n \to 0$, and that $q_n' - q^* = ( 1- \alpha_n) (q^* - q_n)$.
Before insertion the picture is

$$k(q^*) - k(q_n) \to  k'_{left}(q^*) (q^* - q_n),$$

\noindent and after insertion it is

$$k(q'_n) - k(q^*) \to k'_{right}(q^*)(q'_n - q^*) = k'_{right}(q^*) ( 1- \alpha_n) (q^* - q_n).$$

By choosing $n$ large enough, we can guarantee that $k'_{right}(q^*) ( 1- \alpha_n)  >  k'_{left}(q^*)$, and therefore $k(q'_n) - k(q^*) > k(q^*) - k(q_n)$, violating the insertion-accepting property by Lemma \ref{two_vble_char}.
\end{proof}

\begin{remark}  \label{examples} By \citep{dunning2019conservation}, the following are examples of decent scale functions:
$$k_0(q) = {\delta \over 2} q $$

$$k_1(q) =  {\delta \over 2 \pi} \arcsin(2q-1) $$

$$ k_2(q) =  {\delta \over Z(n)} \log{q \over 1-q} $$
 
 $$ k_3(q) =  {\delta \over Z(n)} 
\begin{cases} 
   \log (2 q) & q \leq {1 \over 2} \\
  - \log 2(1-q) & q > {1 \over 2}
       \end{cases}.$$
\end{remark}

\begin{remark}
The unnormalized forms of $k_2, k_3$ (i.e., without the $Z(n)$ term) are also decent. Our conditions in quantile space for the unnormalized forms imply decency in the finite data case for the normalized forms since the function $Z(n)$ is non-decreasing.
\end{remark}

\section{Computations} \label{comps}

\subsection{Piecewise defined functions} \label{piecewise}

\textbf{Gluing.} Suppose $k_l$ and $k_r$ are decent scale functions, and $p \in (0,1)$. Let $k$ denote the function which is $k_l$ on $[0,p]$ and $k_r$ on $(p, 1]$. 
For $k$ to be decent, Proposition \ref{differentiable} implies that $k'_l$ and $k'_r$ must agree at $p$, so one natural approach to gluing is to take a decent scale function $k_r$ (e.g., from the list in Remark \ref{examples}), choose a point $p \in (0,1)$, and let $k_l$ be the best linear approximation to $k_r$ at $p$, i.e., use the function:

 \[  k(q) =  \begin{cases} 
   k'_r(p) (q - p) + k_r(p)  & 0 \leq q \leq  p  \\
    k_r(q) & p < q \leq 1
       \end{cases}
\] 

To show $k$ is decent, by Lemma \ref{one_vble_char} it suffices to show $k((\alpha + (1 - \alpha) q) - k(q)$ and $k((1 - \alpha) q) - k(q)$ are non-increasing on $[0, 1]$. Note if $\alpha + (1 - \alpha) q$ and $q$ are both greater than or equal to $p$, the decency of $k_r$ implies the necessary non-increasing property for $k$ (and similarly via $k_l$ if both are less than or equal to $p$), and similarly for $(1 - \alpha) q$ and $q$. Therefore it suffices to show the non-increasing property for insertions moving $q$ from one side of $p$ to the other, i.e., for $\alpha, p, q$ such that 

\begin{itemize}

\item (left to right) $q \leq p$ and $\alpha + (1 - \alpha) q \geq p$, or

\item (right to left) $(1 - \alpha) q \leq p$ and $q \geq p.$

\end{itemize}

\noindent For the functions $k_r$ we consider, the point $p={1 \over 2}$ is of particular interest since it minimizes the derivative $k_r'$, hence the cluster size for $q \leq p$ is as large as possible.

\textbf{Notation and strategy.} For ease of exposition, we establish common notation for the next three propositions (all concerning the gluing construction). For the case of shifting from left to right, we need to show $g(q) :=  k ( \alpha + (1 - \alpha)q    ) - k( q )$ is non-increasing on the interval defined by $\alpha + (1 - \alpha) q \geq p$ and $q  \leq  p$. For the case of shifting from right to left, we need to show $h(q) :=  k ( (1 - \alpha) q ) -k( q )$ is non-increasing on the interval defined by $(1 - \alpha)q \leq p$ and $q \geq p$. We accomplish this by verifying $g'(q) \leq 0$ and $h'(q) \leq 0$ on the relevant domains. The decency results hold for positive scalar multiples of our scale functions as well (decency is a property of the determined ray), but we leave this implicit for notational simplicity.

\begin{remark}
The construction can be modified in the obvious way to reverse the emphasis on the tails, i.e., using a non-linear scale function on $[0,p]$ and the linear function describing its tangent line at $p$ for $(p,1]$, but we do not explicitly state this variant in our results. The variant with higher accuracy near $q=0$ is reminiscent of a high dynamic range histogram, though the $t$-digest error is still bounded in terms of the quantile $q$ rather than the value of the observation itself.

\end{remark}

\begin{proposition}
For any $p \in (0,1)$, the scale function

 \[ k(q) =  \begin{cases} 
    \frac{q-p}{2 \sqrt{p - p^2}} +   {1 \over 2} \arcsin(2p-1)   & 0 \leq q \leq  p  \\
      {1 \over 2} \arcsin(2q - 1) & p <  q \leq 1% \\
   \end{cases}
\] 

is decent.

\end{proposition}

\begin{proof}
We have
$$g(q) =  {1 \over 2} \arcsin(2  (\alpha + (1 - \alpha)q  ) - 1) -   \frac{q-p}{2 \sqrt{p - p^2}} -   {1 \over 2} \arcsin(2p-1)$$
\noindent and therefore:
$$g'(q) = \frac{  1 - \alpha  }{2 \sqrt{ (\alpha  + (1 - \alpha) q)  -  {( \alpha + (1 -\alpha)q  )}^2 }} - \frac{1}{2 \sqrt{ p - p^2}}$$

\noindent from which it follows that $g'(q)  \leq 0$ is equivalent to

$$(\alpha  + (1 - \alpha) q)(1-q) \geq p (1-p)(1 - \alpha).$$

Since $(\alpha  + (1 - \alpha) q) \geq p$ and $1-q \geq 1 - p$ (as $q \leq p$), the left hand side is greater than or equal to $p (1-p)$. Since $1 - \alpha < 1$, the desired inequality follows.

We calculate:
$$h'(q) = \frac{1 - \alpha}{2 \sqrt{p - p^2}} - \frac{1}{2 \sqrt{q  - q^2}}$$

\noindent from which it follows that $h'(q) \leq 0$ is equivalent to

$$p - p^2 \geq (q - q^2) ( 1 - \alpha)^2.$$

Since $p \geq (1 - \alpha)q$ and $1 -p \geq 1-q$, the left hand side is greater than or equal to  $(1 - \alpha)q(1-q)$, which is greater than the right hand side since $ 1 > 1 - \alpha$.
\end{proof}

\begin{proposition}

For any $p \in (0, 1)$, the scale function 

 \[ k(q) =  \begin{cases} 
    \frac{q-p}{p(1-p)} + \log \frac{p}{1-p} & 0 \leq q \leq  p  \\
    \log \frac{q}{1-q} & p <  q \leq 1% \\
   \end{cases}
\] 

is decent.
\end{proposition}

\begin{proof} We calculate
$$g(q) = \log \frac{\alpha + (1 - \alpha) q}{(1- \alpha)(1-q)} - (  \frac{q-p}{p(1-p)} + \log \frac{p}{1-p}  )$$
\noindent and so
$$g'(q) =   \frac{1 - \alpha}{(1- \alpha)(1-q) (\alpha + (1 - \alpha)q )} - \frac{1}{p(1-p)}.$$

\noindent Therefore $g'(q) \leq 0$ if and only if $(1-q) (\alpha + (1 - \alpha)q ) \geq p(1-p)$. Since $\alpha + (1 - \alpha) q \geq p$ and  $1 -q \geq 1 -p$, the desired inequality follows.

For the other case,
$$h'(q) =  \frac{1 - \alpha  }{ p(1-p)  }  -  \frac{1}{(1-q)q}$$
\noindent and so $h'(q) \leq 0$ when $(1 - \alpha )(1-q)q \leq p(1-p)$.
We have $(1 - \alpha)q \leq p$ and $1-q \leq 1 - p$ and so the result follows.
\end{proof}

\begin{proposition}

For any $p \in [{1 \over 2}, 1)$, the scale function

 \[ k(q) =  \begin{cases} 
 {{q - p} \over {1 - p}}  - \log 2(1-p)  & 0 \leq q \leq  p \\
      - \log 2(1-q) & p <  q \leq 1% \\
   \end{cases}
\] 

is decent.

For any $p \in (0, {1 \over 2})$, the scale function

 \[ k(q) =  \begin{cases} 
  {{q - p} \over p}  + \log 2p  & 0 \leq q \leq  p \\
\log 2q & p < q \leq  {1 \over 2} \\
       - \log 2(1-q) & {1 \over 2}  <  q \leq 1% \\
   \end{cases}
\] 

is decent.
\end{proposition}

\begin{proof}
First we deal with the case the split point $p$ is greater than ${1 \over 2}$. 
We have
$$g(q) = - \log 2 (1 - \alpha) (1 -q) -  {{q - p} \over {1 - p}} + \log 2(1-p)$$
\noindent and so
$$g'(q) = \frac{1}{1-q} - \frac{1}{1-p}.$$
\noindent Since $q \leq p$, we have $\frac{1}{1-q} \leq \frac{1}{1-p}$ and hence $g'(q) \leq 0$ as desired.

For insertions on the other side, we have
$$h(q) =  {{(1 - \alpha) q - p} \over {1 - p}}  - \log 2(1-p)  + \log 2(1- q)$$
\noindent and so
$$h'(q) = {{1 - \alpha} \over {1 - p}}   - {1 \over {1- q}}.$$
\noindent Then $h'(q) \leq 0$ is equivalent to $( 1-\alpha) (1 - q) \leq (1-p)$. Since $q \geq p$, we have $1-q \leq 1-p$ and hence $(1-\alpha)(1-q) \leq 1-p$ as needed.

In the case the split point $p$ is less than ${1 \over 2}$, we have
 \[ g(q) =  \begin{cases} 
 \log 2 (\alpha + (1 - \alpha)q) - {{q-p} \over p} - \log 2p & p \leq \alpha + (1 - \alpha)q \leq {1 \over 2} \\
 - \log 2 (1 - \alpha)(1-q)  - {{q-p} \over p} - \log 2p & {1 \over 2} < \alpha + (1 - \alpha)q
   \end{cases}
\]

\noindent and therefore

\[ g'(q) =  \begin{cases} 
 {{1-\alpha} \over {\alpha + (1-\alpha)q}}  -   {1 \over p}& p < \alpha + (1 - \alpha)q < {1 \over 2} \\
 {1 \over{1-q}} - {1 \over p}   & {1 \over 2} < \alpha + (1 - \alpha)q
   \end{cases} 
\] 

(Note the limit of $g'$ at $\alpha + (1 - \alpha)q = {1 \over 2}$ exists.) For the first branch, $p \leq \alpha + (1 - \alpha)q$ implies $ {1 \over p} \geq {1 \over {\alpha + (1 - \alpha)q }}$, and ${1 \over {\alpha + (1 - \alpha)q }} \geq {{1-\alpha} \over {\alpha + (1 - \alpha)q }}$, so $g'(q) \leq 0$ for these $q$. For the other branch, $q \leq p \leq {1 \over 2}$ implies $q+p \leq 1$, so $1-q \leq p$. Therefore ${1 \over {1-q}} \geq {1 \over p}$ and so $g'(q) \leq 0$ for these $q$ as well.

Next we have
\[ h(q) =  \begin{cases} 
 {{(1 - \alpha)q - p } \over p }  + \log 2p  - \log 2q & p \leq q \leq {1 \over 2} \\
 {{(1 - \alpha)q - p } \over p }  + \log 2p  + \log 2(1-q)  & {1 \over 2} < q  \leq 1
   \end{cases}
\]

\noindent and so 

\[ h'(q) =  \begin{cases} 
 {{1 - \alpha } \over p }   - {1 \over q} & p < q < {1 \over 2} \\
{{1 - \alpha } \over p }   - {1 \over {1-q}} & {1 \over 2} < q  \leq 1
   \end{cases}
\]

(Note the limit of $h'$ at $q = {1 \over 2}$ exists.)
For the first branch, $(1- \alpha) q  \leq p$ implies  ${{1- \alpha} \over p } \leq {1 \over q}$ and so $h'(q) \leq 0$ for these $q$. For the other branch, $h'(q) \leq 0$ is equivalent to $(1 - \alpha)(1 - q) \leq p$. Since $q \geq {1 \over 2}$, we have $1-q \leq q$. Multiplying this inequality by $1-\alpha$ and using $(1- \alpha) q  \leq p$ gives the inequality we need.
\end{proof}

\subsection{Polynomials} \label{polynomials}

\begin{proposition}
For any $B \geq 2$, the scale function $k(q) = q^2 + Bq$ is decent.
\end{proposition}

\begin{proof}
We need to show, for any $\alpha \in (0,1)$, that $g(q) :=  k ( \alpha + (1 - \alpha)q    ) - k( q )$ and $h(q) :=  k ( (1 - \alpha) q ) -k( q )$ are non-increasing on the domain $q \in [0,1]$. Since
$$g(q) = {( \alpha + (1 - \alpha)q  )}^2 + B  (\alpha + (1 - \alpha)q  )  - q^2 - Bq,$$
\noindent we calculate:

\begin{align} \nonumber
g'(q) &= 2 ( \alpha + (1 - \alpha)q  ) (1 - \alpha) + B (1 - \alpha) - 2q - B \\  \nonumber
&= 2q( {(1-\alpha)}^2 - 1) + 2 \alpha ( 1 - \alpha) - B \alpha \\  \nonumber
&= 2q( {(1-\alpha)}^2 - 1) +  \alpha (2( 1 - \alpha) - B) .
\end{align}

Now $B \geq 2$ implies $2( 1 - \alpha) - B < 0$ and so $g'(0) < 0$. Since ${(1-\alpha)}^2 < 1$, $g'$ is decreasing, so $g'$ is negative on $[0,1]$, so $g$ is decreasing on this domain, as desired.

As for $h$, we have
$$h(q) =  {[ (1 - \alpha) q] }^2 + B  (1 - \alpha) q    - q^2 - Bq,$$
and so 
$$h'(q) = 2 {(1 - \alpha)}^2 q + B(1 - \alpha) - 2q - B = 2q( {(1 - \alpha)}^2 -1)  - B \alpha.$$
\noindent Now $h'(0) = - B \alpha < 0 $ and $h'$ is decreasing, so $h'$ too is negative on $[0,1]$, so $h$ is decreasing on this domain, as desired.
\end{proof}

More generally we have the following.

\begin{proposition}
For any $n > 0$, there exists $B_n > 0$ such that for $B \geq B_n$, the scale function $k(q) = q^n + Bq$ is decent.
\end{proposition}

\begin{proof}
We find conditions on $B$ guaranteeing that $g(q) :=  k ( \alpha + (1 - \alpha)q    ) - k( q )$ and $h(q) :=  k ( (1 - \alpha) q ) -k( q )$ are non-increasing on the domain $q \in [0,1]$, for any $\alpha \in (0,1)$. We have:
\begin{align} \nonumber
g'(q) &= (1 - \alpha) [ n {(\alpha + (1 - \alpha)q)}^{n-1} + B ] - (n q^{n-1} + B) \\
\nonumber & = n [ (1 - \alpha) {( q + \alpha ( 1- q))}^{n-1}  - q^{n-1}] - B \alpha
\end{align}

Now $a(\alpha, q) := (1 - \alpha) {( q + \alpha ( 1- q))}^{n-1}  - q^{n-1} = (1 - \alpha)( q^{n-1} + p(\alpha, q)) - q^{n-1}$, where $ p(\alpha, q)$ is a polynomial divisible by $\alpha$ and with no $B$-dependence. Hence $a(\alpha, q) = - \alpha q^{n-1} + (1 - \alpha) p (\alpha, q)$ is divisible by $\alpha$ (say $a = \alpha a_1$) and we can write

$$g'(q) = n \alpha a_1(\alpha, q) - B \alpha = \alpha (n a_1(\alpha, q) - B).$$

\noindent Now $a_1$ is a polynomial in $\alpha, q$, in particular has a maximum $M$ on $[0,1] \times [0,1]$. Choosing $B$ larger than $M$ implies $g'(q) < 0$ as desired, so we can choose $B_n$ to be anything larger than $M$ (which depends only on $n$).

The analysis of $h(q)$ is somewhat simpler. We calculate:
\begin{align} \nonumber
h'(q) &= (1 - \alpha) [n {((1 - \alpha) q )}^{n-1} + B] - (n q^{n-1} + B) \\
\nonumber & = n q^{n-1} ( {(1 - \alpha)}^n - 1) - B \alpha
\end{align}

\noindent Now $ n q^{n-1} \geq 0$ and $( {(1 - \alpha)}^n - 1) < 0 $, so the first term is non-positive, and $B \alpha \geq 0$, so $h'(q) \leq 0 $ as desired. \end{proof}

Combining various polynomials using the convex cone property (Lemma \ref{cone}), we can generate lots of decent scale functions. The utility of this construction is somewhat unclear, as the linear term dominates more with larger $B$. The decency of certain polynomials also opens the possibility of extending the gluing construction from linear approximations to higher degree Taylor polynomials.

\section{Empirical results} \label{results}

This section summarizes the results of 100 runs of constructing a $t$-digest on one million samples from a uniform distribution, for different scale functions.
The main goal is to understand empirically the effect of the scale functions discussed in Section \ref{comps}, especially the gluing construction applied to the familiar scale functions. In all cases we set the compression parameter $\delta = 100$ and perform a compression (so the digest is ``fully merged") before calculating quantiles. 
We follow the conventions of \citep{dunning2019computing} (see also Remark \ref{examples}). 
For the piecewise defined functions, we glue at the point $p={1 \over 2}$. For $q \geq {1 \over 2}$ we use the size bounds of \citep{dunning2019size} and for $q \leq {1 \over 2}$ we bound by the reciprocal of the slope of the line; for $k_1$ these differ by higher order terms in the normalizing/compression factors.

The error is the absolute value of the difference between the cumulative distribution function evaluated at the estimate of quantile $q$ and $q$ itself, and appears in the leftmost panel. The normalized error divides this quantity by $\min(q, 1-q)$ and appears in the center panel.
For the error plots, the whiskers range from the 5th to 95th percentile of the 100 runs, the boxes cover the interquartile range, the orange line is the median, and the horizontal axis is the following transformation of quantile space:
{ \[  \begin{cases} 
   \log_{10}(q)   & 0 < q <  {1 \over 2}  \\
   0 & q = {1 \over 2} \\
     - \log_{10}(1 - q)  & {1 \over 2}<  q < 1% \\
   \end{cases}
\] }
Note the horizontal axes have the same interpretation in each figure, but the vertical axes vary. The rightmost panel is a histogram of centroid counts over the 100 runs. 
Implementations of the asymmetric scale functions and code for generating the data and plots are available at \citep{asymmrepo} (a fork of \citep{digestrepo}).

\subsection{AVL tree results}

This subsection compares the different scale functions for $t$-digests using the AVLTree variant in the Java implementation \citep{digestrepo}.

\begin{figure}[H]
  \caption{Errors and centroid counts for the scale function $k_0$ (first row; a baseline) and for the quadratic polynomial scale function $k_{quadratic} = { \delta \over 6 } (q^2 + 2q)$ (second row). For this coefficient choice, the resulting function maps $[0,1]$ to $[0,{ \delta \over 2 }]$, as does $k_0$. Both use AVLTree.
  }
  \centering
\includegraphics[width=\textwidth]{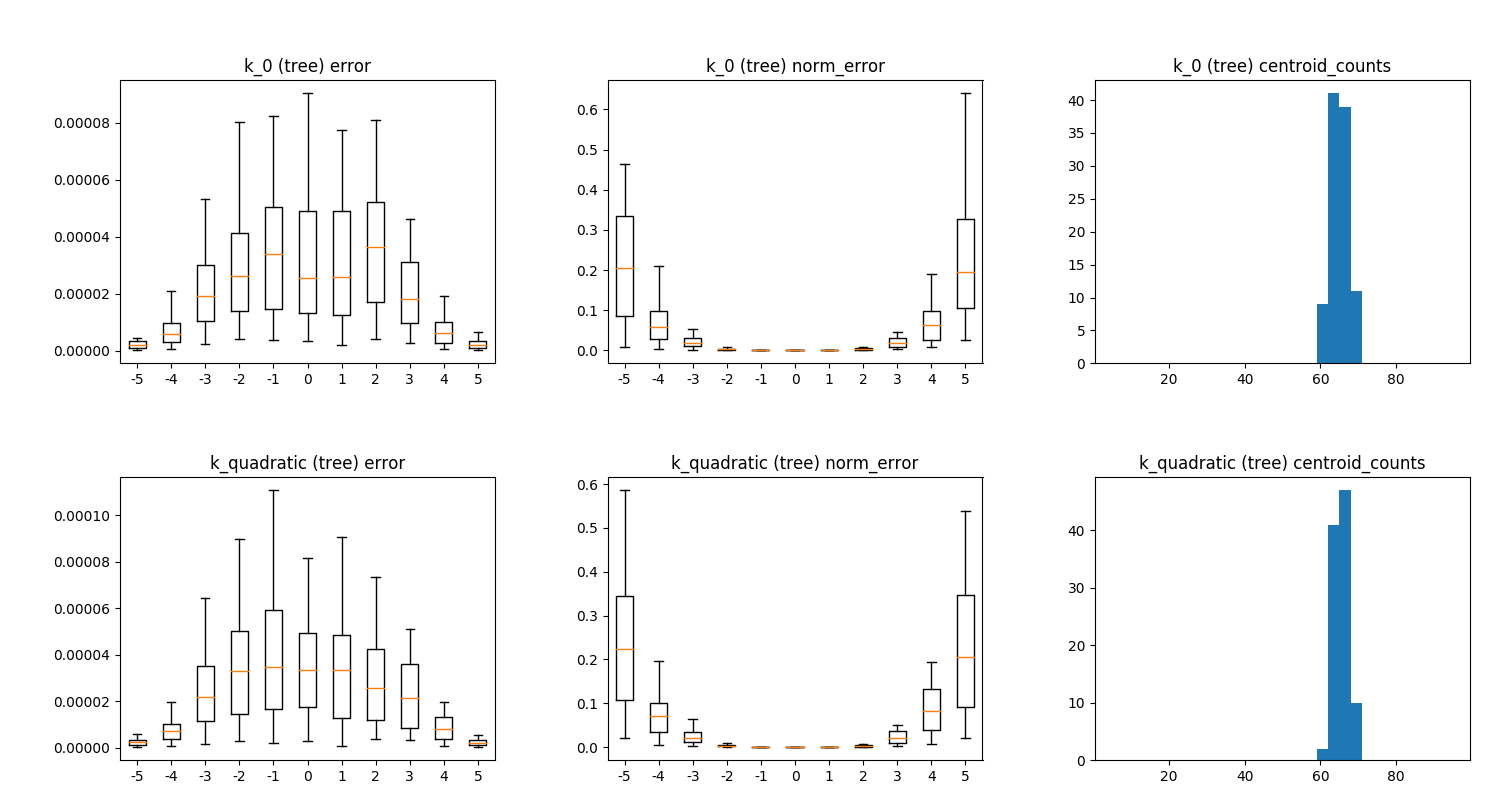}
\end{figure}

\begin{figure}[H]
  \caption{Errors and centroid counts for the usual (first row) and glued (second row) variants of the scale function $k_1$. Both use AVLTree.}
  \centering
\includegraphics[width=\textwidth]{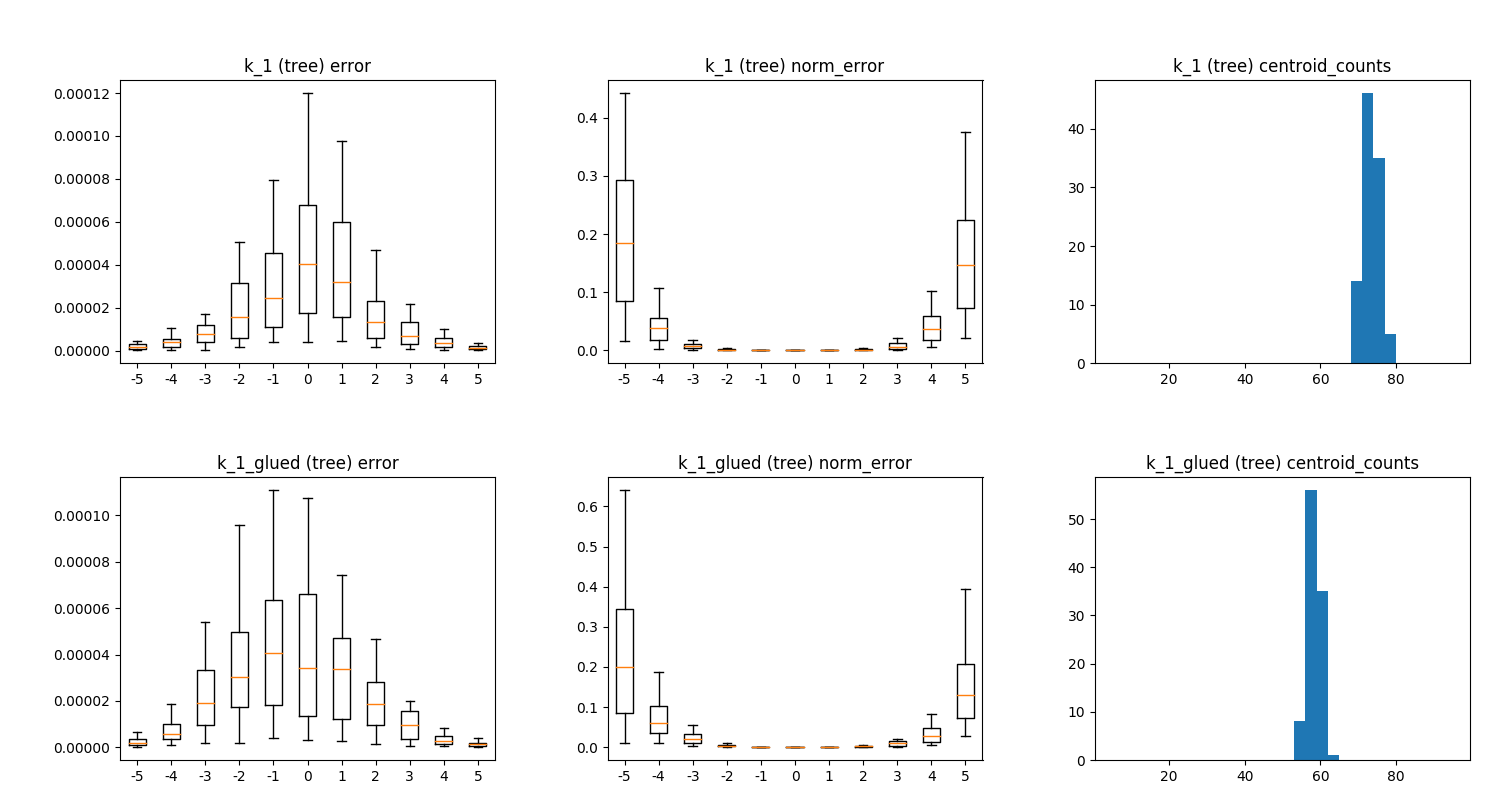}
\end{figure}

\begin{figure}[H]
   \caption{Errors and centroid counts for the usual (first row) and glued (second row) variants of the normalized scale function $k_2$. Both use AVLTree.}
  \centering
\includegraphics[width=\textwidth]{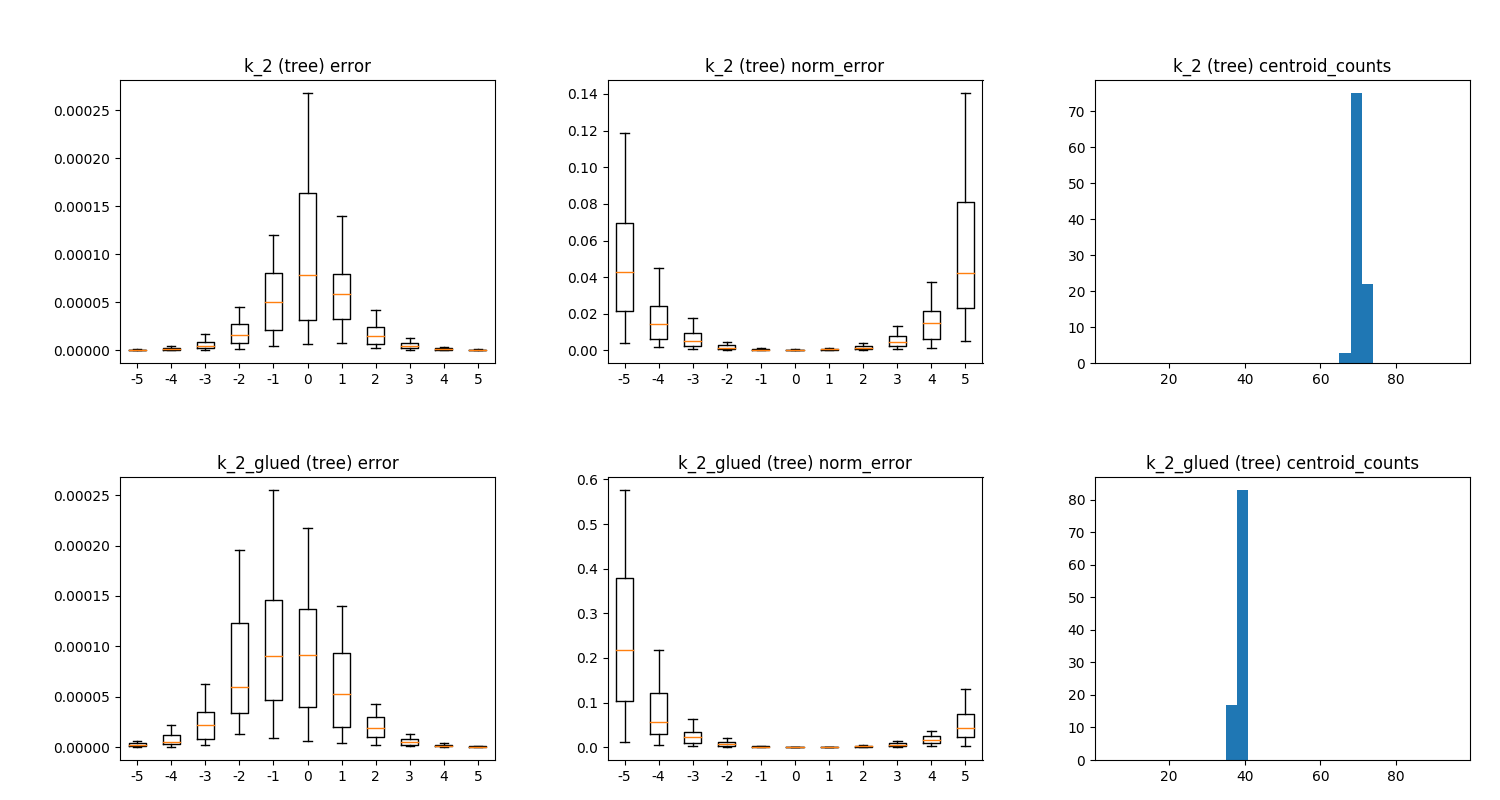}
\end{figure}

\begin{figure}[H]
    \caption{Errors and centroid counts for the usual (first row) and glued (second row) variants of the normalized scale function $k_3$. Both use AVLTree.}
  \centering
\includegraphics[width=\textwidth]{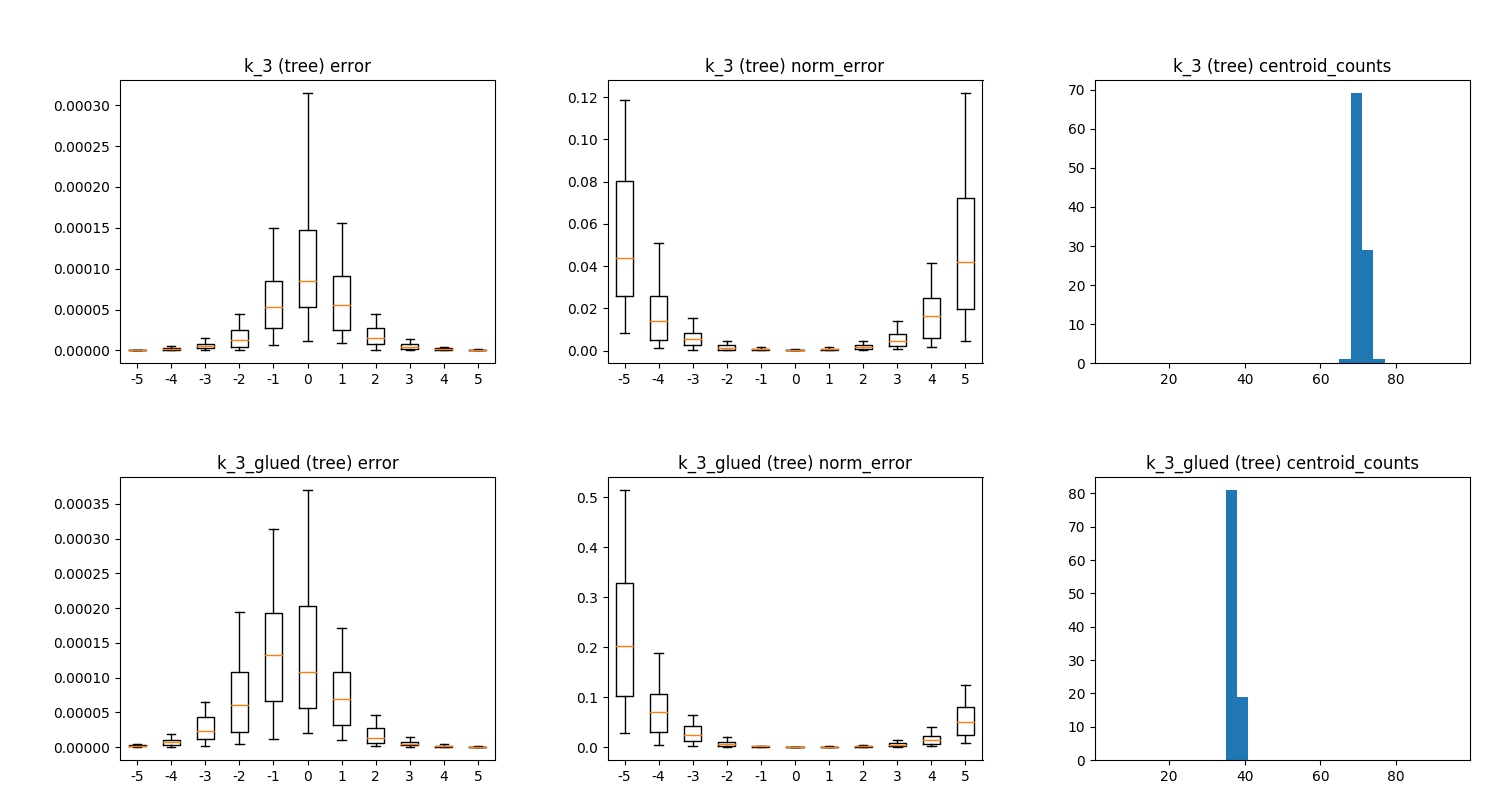}
\end{figure}

\textbf{Discussion.} In all cases the glued variant of $k_i$ has the error profile of $k_i$ for $q \geq {1 \over 2}$ and that of $k_0$ (linear function, uniform cluster sizes) for $q \leq {1 \over 2}$, as expected.
The reduction in number of centroids is more dramatic for $k_2$ and $k_3$ than it is for $k_1$ due to the normalizing term $Z(n)$ appearing in the linear halves of $k_2$ and $k_3$. This reduction describes, to first order, the memory savings of the asymmetric (glued) variant over the usual symmetric one. We have not investigated quantitatively the computational advantage, but roughly speaking, half (when gluing at $p={1 \over 2}$) of the transcendental scale function evaluations are replaced by evaluation of a simple linear function.

\subsection{Merging digest results}

This subsection compares the different scale functions for $t$-digests using the MergingDigest variant in the Java implementation \citep{digestrepo}. We have made two minor changes to the main implementation in  \citep{asymmrepo}. First, we set ``useAlternatingSort" to false, so that we do not alternate between upward and downward merge passes.
Alternating seems to interact poorly with asymmetric scale functions; when set to true, the digests using asymmetric scale functions have too few centroids. 
Second, we have added more padding to the underlying arrays; the amount of fudge required seems to depend on the number of samples processed.

\begin{figure}[H]
  \caption{Errors and centroid counts for the scale function $k_0$ (first row; a baseline) and for the quadratic polynomial scale function $k_{quadratic} = { \delta \over 6 } (q^2 + 2q)$ (second row). For this coefficient choice, the resulting function maps $[0,1]$ to $[0,{ \delta \over 2 }]$, as does $k_0$. Both use MergingDigest.
 For $k_0$, the unusual errors at $q=0.01$ and $q=0.99$ (several times the error observed with the AVLTree implementation) seem to be related to the compression parameter (perhaps via inaccuracy near the boundary between clusters); these ``bumps" move to $q=0.001$ and $q=0.999$ with $\delta=1000$. The asymmetric $k_{quadratic}$ improves the error at $q=0.99$ at the expense of introducing more centroids.
  }
  \centering
\includegraphics[width=\textwidth]{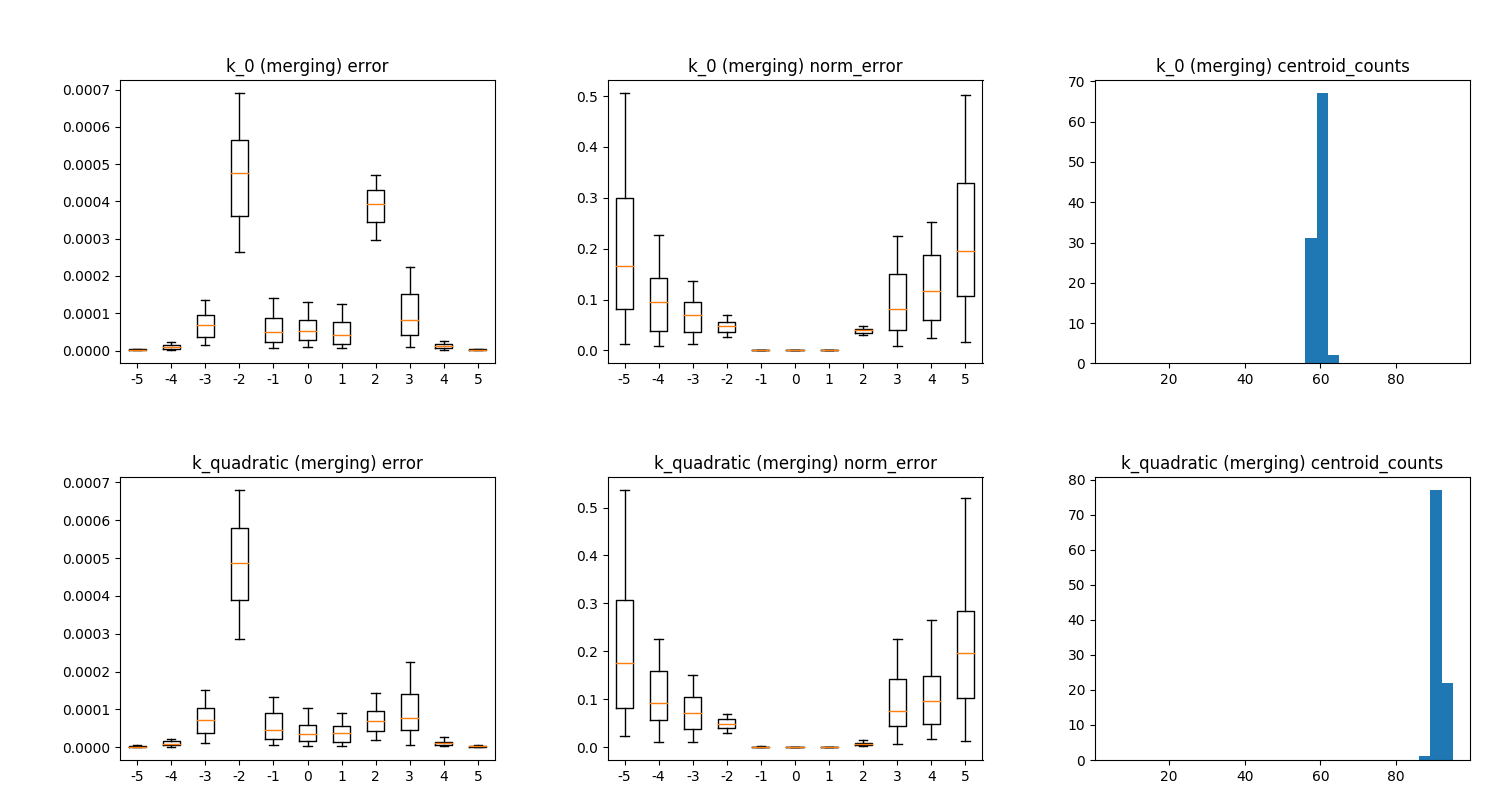}
\end{figure}

\begin{figure}[H]
 \caption{Errors and centroid counts for the usual (first row) and glued (second row) variants of the scale function $k_1$. Both use MergingDigest.
The glued variant of $k_1$ has the error profile of $k_1$ for $q \geq {1 \over 2}$ and that of $k_0$ for $q \leq {1 \over 2}$, including the unusual error at $q=0.01$. The unexpected asymmetry of the errors for $k_1$ disappears when setting ``useAlternatingSort" to true.
    }
  \centering
\includegraphics[width=\textwidth]{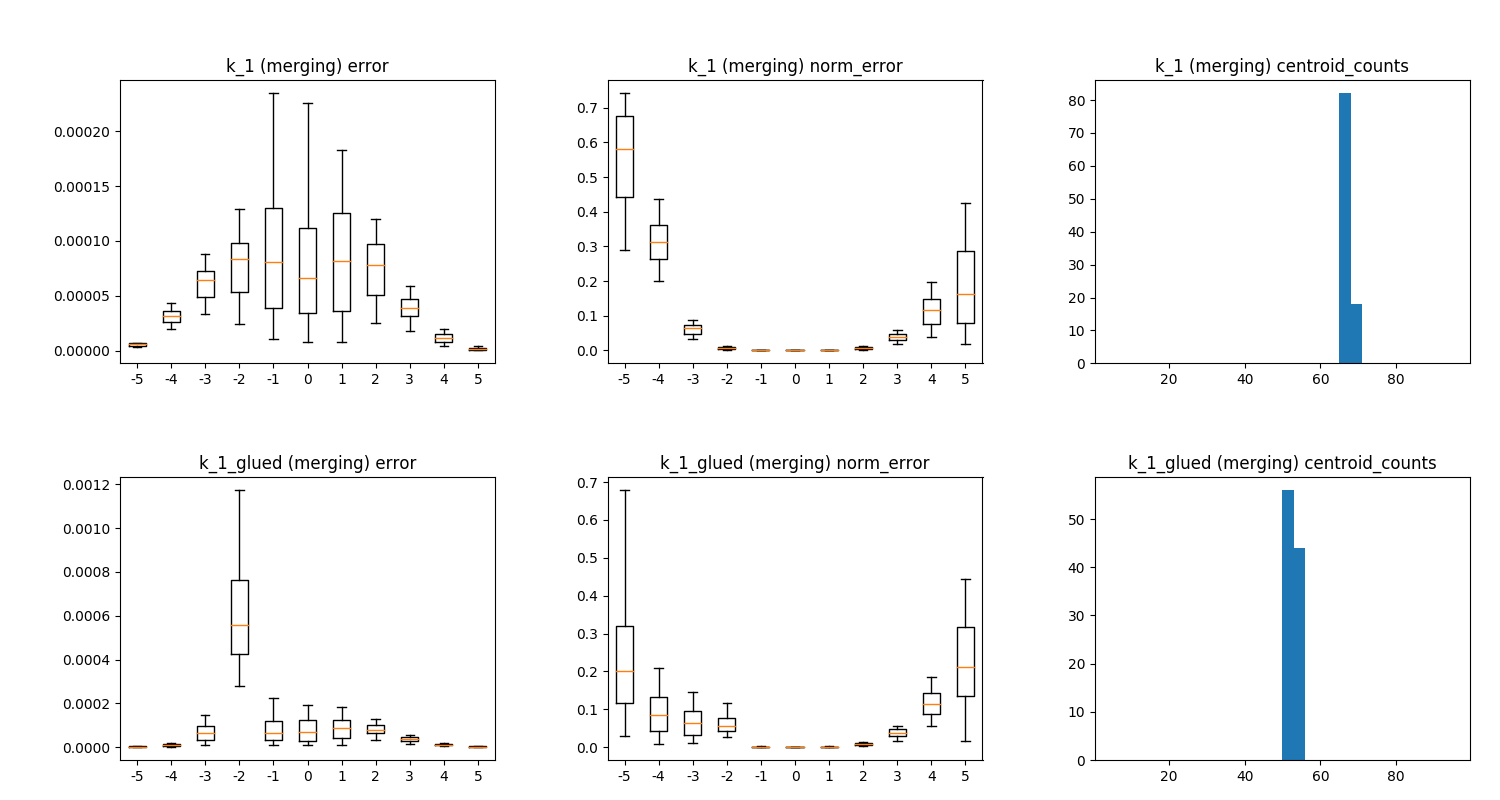}
\end{figure}

\begin{figure}[H]
    \caption{Errors and centroid counts for the usual (first row) and glued (second row) variants of the normalized scale function $k_2$. Both use MergingDigest.
The glued variant of $k_2$ has the error profile of $k_2$ for $q \geq {1 \over 2}$ and that of $k_0$ for $q \leq {1 \over 2}$, except some of the unusual error for $k_0$ at $q=0.01$ seems to have shifted to $q=0.1$ (perhaps due to more effective compression for $q \leq {1 \over 2}$).
 }
  \centering
\includegraphics[width=\textwidth]{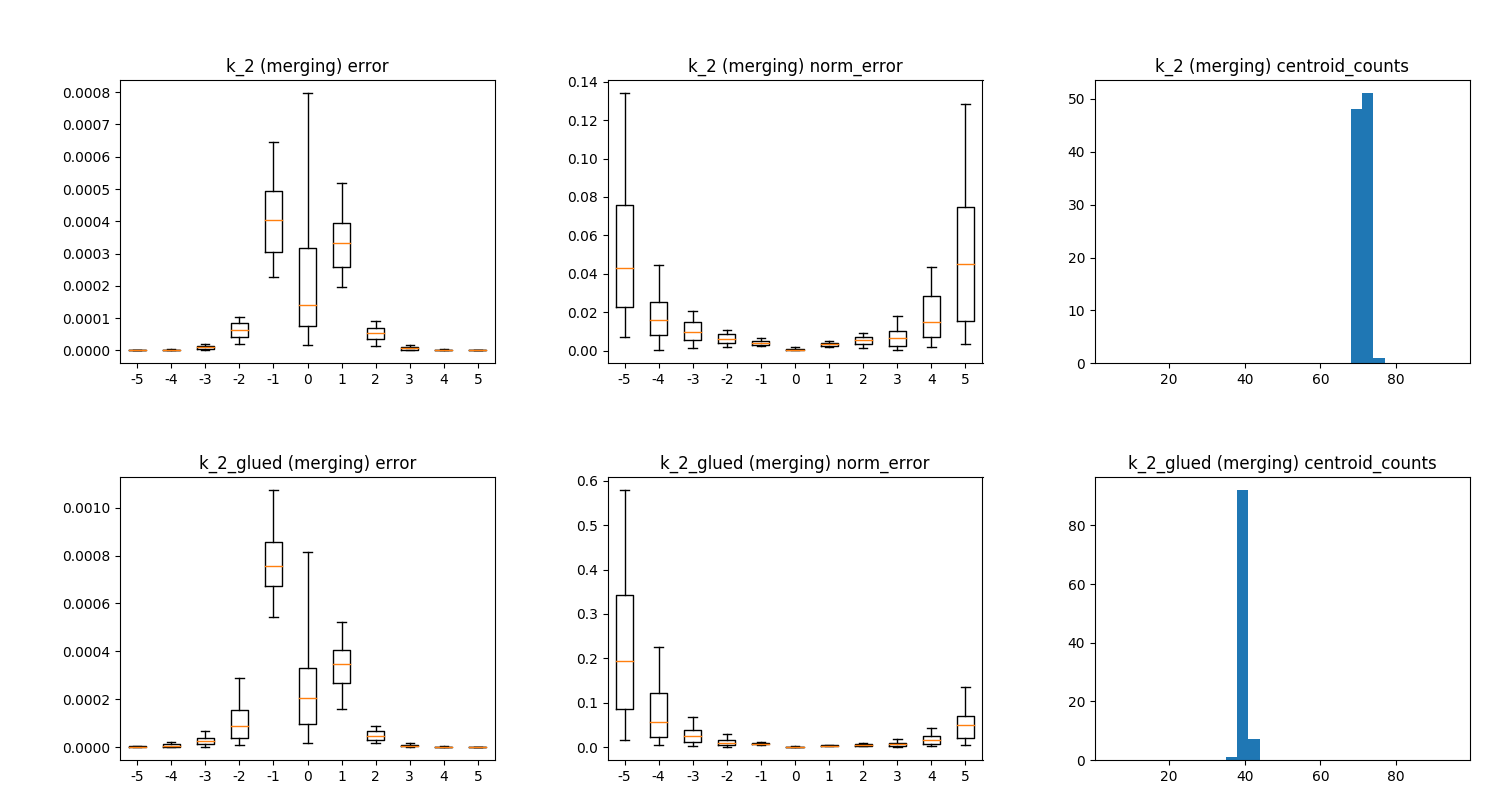}
\end{figure}

\begin{figure}[H]
\caption{Errors and centroid counts for the usual (first row) and glued (second row) variants of the normalized scale function $k_3$. Both use MergingDigest.
The glued variant of $k_3$ has the error profile of $k_3$ for $q \geq {1 \over 2}$ and that of $k_0$ for $q \leq {1 \over 2}$, except some of the unusual error for $k_0$ at $q=0.01$ seems to have shifted to $q=0.1$ (perhaps due to more effective compression for $q \leq {1 \over 2}$).}
  \centering
\includegraphics[width=\textwidth]{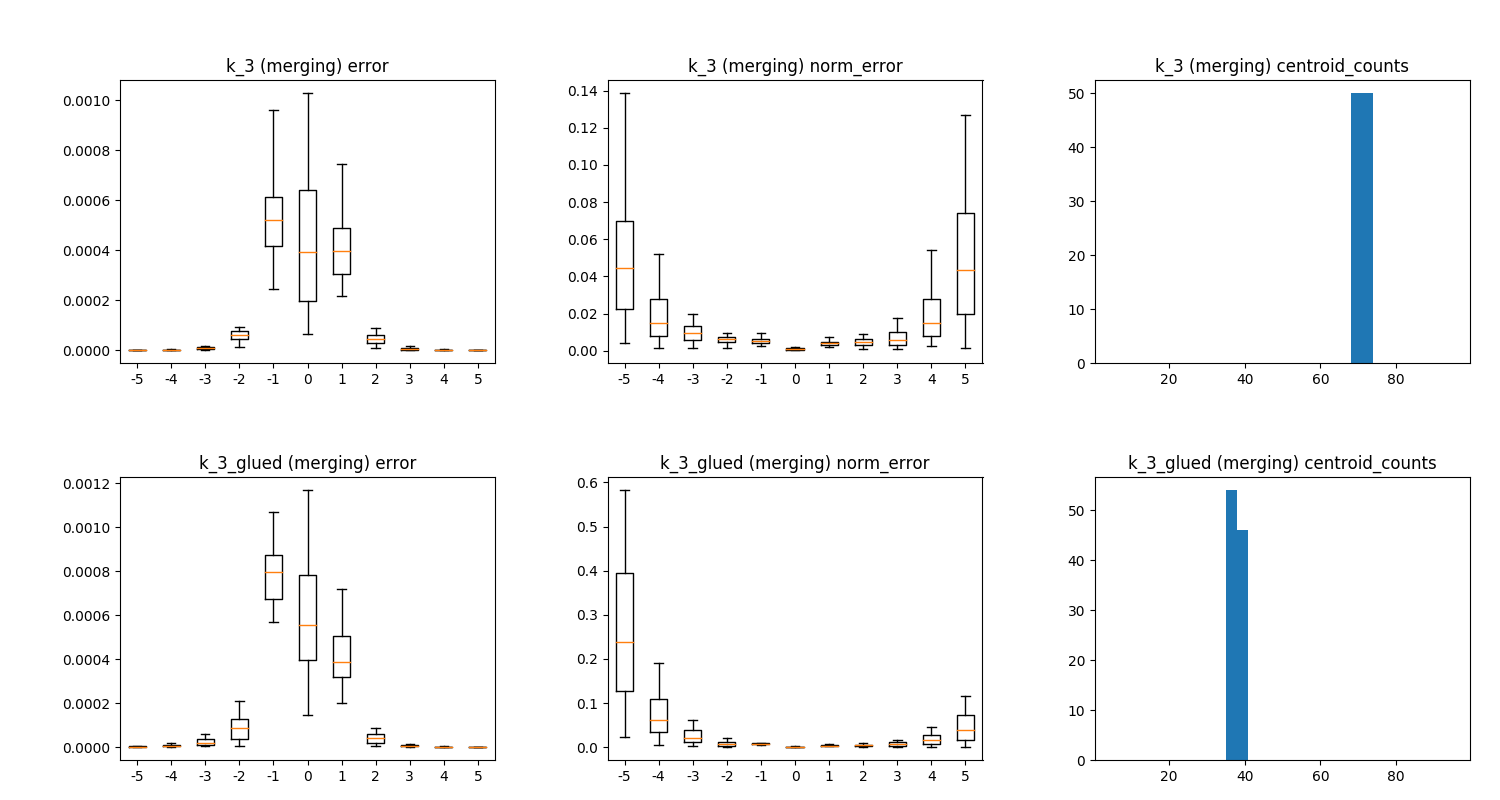}
\end{figure}

\textbf{Acknowledgments.} It is a pleasure to thank engineering and management at SignalFx for their encouragement and support during the preparation of this paper, and to thank Matthew Pound for many interesting discussions on this topic.

\bibliography{tdigest_bib}{}
\bibliographystyle{plainnat}

\end{document}